\documentclass[12pt]{article}

\usepackage{amsmath,amssymb,amsthm, units, color,wasysym}

\usepackage{graphicx}

\usepackage[all]{xy}

\usepackage{tikz}

\newcommand\pair[2]{({#1},{#2})}

\newtheorem{definition}{Definition}[section]

\newtheorem{lemm}{Lemma}[section]

\newtheorem{theo}{Theorem}[section]

\def\<{\left <}

\def\>{\right >}

\def\({\left (}

\def\){\right )}

\DeclareSymbolFont{AMSb}{U}{msb}{m}{n}

\DeclareMathSymbol{\N}{\mathbin}{AMSb}{"4E}

\DeclareMathSymbol{\Z}{\mathbin}{AMSb}{"5A}

\DeclareMathSymbol{\R}{\mathbin}{AMSb}{"52}

\DeclareMathSymbol{\Q}{\mathbin}{AMSb}{"51}

\DeclareMathSymbol{\I}{\mathbin}{AMSb}{"49}

\DeclareMathSymbol{\C}{\mathbin}{AMSb}{"43}

\newcommand{\weg}[1]{}

\title{A geometric protocol for cryptography with cards}

\author{Andr\'es Cord\'on--Franco, Hans van Ditmarsch, \\ David Fern\'andez--Duque, and Fernando Soler--Toscano\thanks{Email: {\tt \{acordon,hvd,dfduque,fsoler\}@us.es}. Affiliation: University of Sevilla, Spain. Affiliation Hans v.D.: LORIA, France, and (as research associate) IMSc, India.}}

\date{\today}

\begin{document}
\maketitle
\begin{abstract}
In the {\em generalized Russian cards problem,} the three players Alice, Bob and Cath draw $a,b$ and $c$ cards, respectively, from a deck of $a+b+c$ cards. Players only know their own cards and what the deck of cards is. Alice and Bob are then required to communicate their hand of cards to each other by way of public messages. The communication is said to be {\em safe} if Cath does not learn the ownership of any specific card; in this paper we consider a strengthened notion of safety introduced by Swanson and Stinson which we call {\em $k$-safety.}

An elegant solution by Atkinson views the cards as points in a finite projective plane.  We propose a general solution in the spirit of Atkinson's, although based on finite vector spaces rather than projective planes, and call it the `geometric protocol'. Given arbitrary $c,k>0$, this protocol gives an informative and $k$-safe solution to the generalized Russian cards problem for infinitely many values of $(a,b,c)$ with $b=O(ac)$. This improves on the collection of parameters for which solutions are known. In particular, it is the first solution which guarantees $k$-safety when Cath has more than one card.

\end{abstract}

\section{Introduction}

The {\em generalized Russian cards problem} \cite{albertetal:2005,swanson:2012} is a family of combinatorial puzzles about secure secret-sharing between card players. It is parametrized by a triple of natural numbers $(a,b,c)$, which we call its {\em size}, and can be stated as follows:
\begin{quote}
{\bf The generalized Russian cards problem} 

\medskip

{Alice, Bob and Cath each draw $a,b$ and $c$ cards, respectively, from a deck containing a total of $a+b+c$. All players know which cards were in the deck and how many of them the other players drew, but may only see the cards in their own hand.

Alice and Bob want to know exactly which cards the other holds. Meanwhile, they do not want for Cath to learn who holds any card whatsoever, aside of course from her own cards.

However, they may only communicate by making true, clear, public announcements, so that Cath can learn all the information that they exchange.

Can Alice and Bob achieve this?}
\end{quote}
The solutions to this problem are given by {\em protocols.} The eavesdropper Cath is able to hear all communications, and Alice and Bob (and Cath) are aware of that. In such information-exchanging protocols between Alice and Bob, Cath typically acquires quite a bit of data, but not enough to be able to deduce any secrets. Many protocols have been proposed \cite{albertetal:2005,cordonetal:2012,swanson:2012}, and cryptography based on card deals is also investigated in other settings \cite{fischeretal:1996,mizukietal:2002,stiglic:2001}, yet it remains unknown exactly for which triples $(a,b,c)$ the problem can be solved.

The goal of this article is to present a new protocol based on finite linear algebra, inspired by a protocol based on projective geometry by Atkinson that was reported in \cite{albertetal:2005}, and that we present in Section \ref{ExampleSec}. This {\em geometric protocol} solves the generalized Russian cards problem in many instances that were previously unsolved. Also, it employs a stronger notion of security than is often considered, which we call {\em $k$-safety} and is equivalent to {\em weak $k$-security} in \cite{swanson:2012}.

Some general assumptions are needed to make the problem precise, which we shall formalize in Section \ref{sec.some}, along with the notion of {\em protocol}. First, the cards are dealt beforehand in a secure phase which we treat as a black box and gives the players no information about others' cards. The agents have no communication before this phase and cannot share secrets (such as private keys). Second, the agents have unlimited computational capacity. This means, on one hand, that solutions via encryption are not valid, provided they are vulnerable to cryptanalysis (regardless of its difficulty). It also means that we shall not be concerned with the feasibility of agents' strategies, although this is certainly an interesting line for future inquiry. Third, all strategies are public knowledge, keeping with Kerkhoffs' principle \cite{kerckhoffs:1883}.

\subsection{Known solutions}

Many solutions consist of two announcements. First, Alice announcess of a number of possible hands she may hold; then, Bob discloses Cath's cards (or other equivalent information). Atkinson's solution is for the case $(3,3,1)$ and consists of viewing the cards as points in a projective plane, in such a way that Alice holds a line. She then announces the seven lines of that plane, i.e., seven triples of cards. This protcol is presented in \cite{albertetal:2005}, along with many incidental results, and will be discussed further in Section \ref{ExampleSec}. The size $(3,3,1)$ was first considered by Kirkman \cite{kirkman:1847}, who suggests a solution using a {\em design}, a collection of subsets of a given set that satisfies certain regulaties \cite{stinson:2004}. The design consists of seven triples --- and accidentally these are the lines that form the projective geometric plane.

A possibly better-known solution for $(3,3,1)$ is to number the cards $0,\hdots,6$, after which Alice and Bob announce the sum of their cards modulo 7 (i.e., modulo the number of cards in the pack). This was the proposed solution when the problem appeared in the Moscow Mathematics Olympiad in 2000, from which it earned its current name \cite{makarychevs:2001}; a very similar approach works in any case where Cath holds only one card and $a,b>2$ \cite{cordonetal:2012}. A protocol of three announcements for $(4,4,2)$ is reported in \cite{threesteps}, and a four-step protocol for $c=O(a^2)$ and $b=O(c^2)$ is presented in \cite{colouring}. The last two are notable because they provide solutions for $c > 1$ and neither can be solved with two-announcement protocols.

Stronger notions of security are studied by Swanson and Stinson in \cite{swanson:2012}. There, a distinction is made between {\em weak} and {\em perfect} security; in perfectly secure protocols, Cath does not acquire any probabilistic information about the ownership of any specific card. They also introduce {\em $k$-secure protocols,} where tuples of at most $k$ cards are considered simultenously. They then characterize the perfect $(k-1)$-secure solutions for sizes $(a,b,a-k)$ and show that $k=a-1$ and $c=1$. These notions will be revisited in Section \ref{sec.some}, but we remark that the protocol we shall present is weakly $k$-secure. 

\subsection{Plan of the paper} In Section \ref{ExampleSec} we give an informal description of the protocol. Section \ref{sec.some} formalizes our model of security and our notion of protocol, so that in Section \ref{SecGeoProt} we may give a rigorous specification. The protocol depends on several parameters and is only executable under certain constraints; Section \ref{SecCompPat} computes solutions to these constraints.

\subsection{Finite geometry} We assume some basic familiarity with finite fields and finite geometry; these are covered in texts such as \cite{lidl1997} and \cite{dembowski1997}, respectively.

Throughout the paper, $p$ will denote a prime or a power of a prime, and $\mathbb F_p$ the field with $p$ elements. If $d$ is any natural number, $\mathbb F^d_p$ denotes the vector space of dimension $d$ over $\mathbb F_p$. Given sets $U,V\subseteq\mathbb F^d_p$ we write $\langle U\rangle$ for the span of $U$ (i.e., the set of all linear combinations of elements of $U$), and we write $U+ V$ for the set $\{u+v:u\in U\text{ and }v\in V\}$. We may write $x+ V$ instead of $\{x\}+ V$, and similarly $\langle x,U\rangle$ instead of $\langle \{x\}\cup U\rangle$. By a {\em hyperplane} we mean any set of the form $x+V$ where $V$ is a subspace of dimension $d-1$, and two hyperplanes $X,Y$ are {\em parallel} if $X\not=Y$ but there is a vector $x$ such that $X=x+Y$.
 
Recall that $|\mathbb F^d_p|=p^d$, where $|X|$ denotes the cardinality of $X$. Moreover, if $U\not =V$ are hyperplanes, then $U$ has exactly $p^{d-1}$ elements, while $|U\cap V|\leq p^{d-2}$. Parallel hyperplanes have empty intersection.

\section{Motivating examples}\label{ExampleSec}

Let us begin by presenting Atkinson's solution for the case $(3,3,1)$. In this setting, Alice and Bob each draw three cards from a deck of seven cards, while Cath gets the remaining card. The claim is that Alice and Bob can communicate their cards to each other by way of public announcements, without informing Cath of any of their cards. First, Alice announces that her hand is a line in a projective plane consisting of seven points (cards). Or, to be precise, Alice assigns a point in the projective plane to each card in such a way that her own hand forms a line, and then announces ``The hand I hold is one of the following\textellipsis,'' after which she proceeds to list every set of cards which corresponds to a line. Then, to conclude the protocol, Bob announces Cath's card.

\begin{figure}[h]
\begin{center}
\begin{tikzpicture}
\node[label=below:0] (0) at (0,0) {$\bullet$};
\node[label=below:1] (1) at (2,0) {$\bullet$};
\node[label=below:2] (2) at (4,0) {$\bullet$};
\node[label=right:4] (3) at (3,1.73) {$\bullet$};
\node[label=above:5] (4) at (2,3.46) {$\bullet$};
\node[label=left:6] (5) at (1,1.73) {$\bullet$};
\node (6) at (2,1.15) {$\bullet$};
\node (6node) at (2.25,1.6) {3};
\draw (0,0) -- (4,0) -- (2,3.46) -- (0,0);
\draw (0,0) -- (3,1.73);
\draw (1,1.73) -- (4,0);
\draw (2,0) -- (2,3.46);
\draw (2,1.15) circle (1.15);
\end{tikzpicture}
\end{center}
\caption{Alice holds a line in the 7-point projective plane}
\label{projplane}
\end{figure}
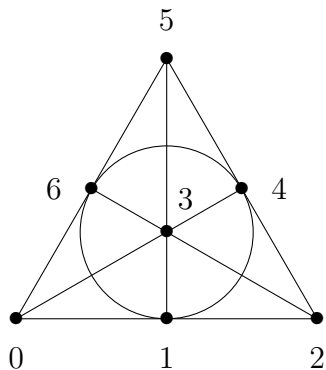

Why does this work? Suppose that the cards are numbered $0, 1, \hdots, 6$, that Alice holds the cards $0$, $1$ and $2$, Bob holds $3$, $4$, and $5$, and therefore Cath holds $6$. Alice announces: ``My cards form a line in the projective plane whose lines are $012$, $034$, $056$, $135$, $146$, $236$, and $245$.'' (See Figure \ref{projplane}.) Bob then announces: ``Cath holds $6$.'' After Alice's announcement, Cath, who holds card $6$, can eliminate from the seven triples the ones containing $6$: $056$, $146$, and $236$. The remaining hands are: $012$, $034$, $135$ and $245$. Cath therefore cannot deduce that Alice has $0$, because $135$ is a possible hand of Alice. She also cannot deduce that Alice does not have $0$, because Alice's actual hand $012$ is also a possible hand. And so on, for all possible cards of Alice. Also---and this is important---for any other deal of cards in which Alice can truthfully make this announcement we can repeat this exercise, e.g., also when Alice holds $012$ and Cath $4$, also when Alice holds $135$ and Cath $0$, and so on. 

Meanwhile, Bob learns Alice's cards from her announcement, because all but $012$ contain either a $3$, a $4$, or a $5$. Again, we have to do this for all card triples $xyz$, not just for Bob's actual hand $345$, and also for all seven possible hands of Alice and all hands Bob can have in that case. After Alice's announcement Bob therefore always knows Cath's card and can announce it, from which Alice also learns the entire deal of cards. 

Let us give an informal account of the geometric protocol to motivate the formal description later. It is similar to Atkinson's, but there are three main differences: 
\begin{enumerate}
\item Projective spaces are replaced by vector spaces; in fact, the original protocol is not dependent on any particular property of projective spaces that they do not share with vector spaces.
\item Rather than considering exclusively planes, we work over spaces of arbitrary dimension, so that in general Alice arranges her cards on a hyperplane.
\item Alice may have more cards than fit on a single hyperplane. Thus she shall arrange her cards on several parallel hyperplanes.
\end{enumerate} 

The protocol works as follows. Fix a size $(a,b,c)$ such that there
are integers $d,k>0$ and a prime power $p$ with $a=kp^{d}$ and
$a+b+c=p^{d+1}$. Suppose that $(A,B,C)$ has been dealt, and that $D = A \cup B \cup C$. Alice chooses
a map $f:D\to\mathbb F^{d+1}_p$, such that $A$ is the disjoint union of
$k$ parallel hyperplanes. Then, she announces the set $\mathcal A$ of
all hands $X$ with $a$ elements such that $f(X)$ has this form; in other words:
such that $X$ maps to the disjoint union of $k$ parallel hyperplanes. 

In Section \ref{SecGeoProt}, in Definition \ref{defprot}, we shall give a more rigorous definition of this protocol, and we then also show this protocol to be  $k$-safe, provided the parameters satisfy certain constraints; $k$-safe means that for any $k$ cards not held by Cath, she cannot learn whether Alice all holds them (i.e., if she holds all those cards or does not hold some of them). In that sense, the above projective plane solution for $(3,3,1)$ is $1$-safe.

But first, let us focus on a specific instance to see the protocol in
action. Consider a card deal of size $(8,6,2)$, so that there are 16 cards, of which Alice holds eight, Bob six and Cath two. For this example we will work over $\mathbb F^2_4$; note that the field $\mathbb F_4$ has elements $\{0,1,\alpha,\alpha^2\}$ where $\alpha$ is a root of $x^2+x+1$. In order to execute the geometric protocol, Alice first announces that her cards are two parallel lines in $\mathbb F^2_4$; after this, Bob informs Alice of Cath's cards. 

Let us see why this protocol works. The deck $D$ may be any set with 16 elements. The only thing that matters is a bijection $f:D\to\mathbb F^2_4$, where we represent $\mathbb F^2_4$ as $\{\pair ij \mid i,j\in \{0,1,\alpha,\alpha^2\} \}$.

  \begin{figure}[htbp!]
  \centering
  \includegraphics[width=3cm]{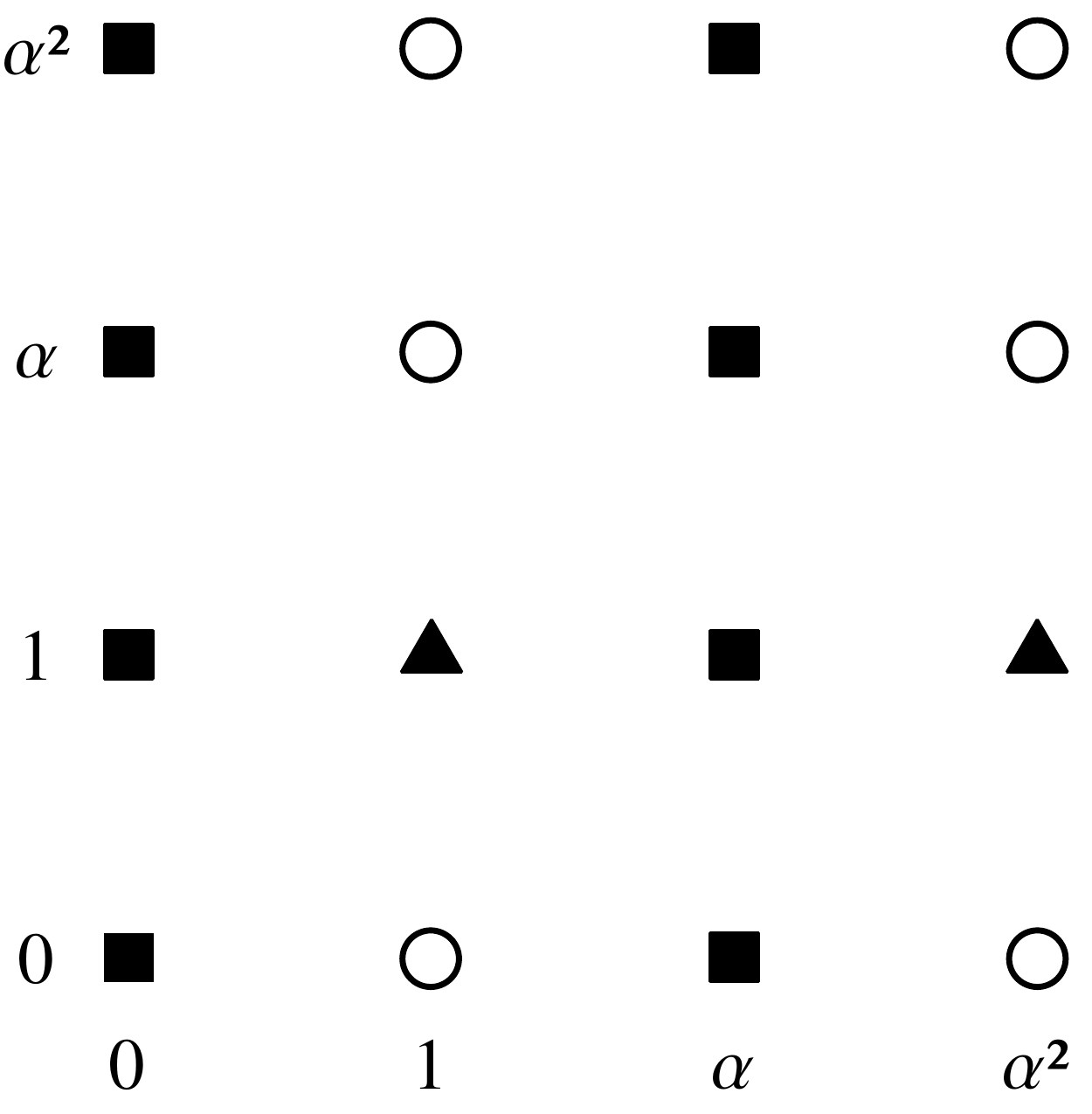}\hfill
  \includegraphics[width=3cm]{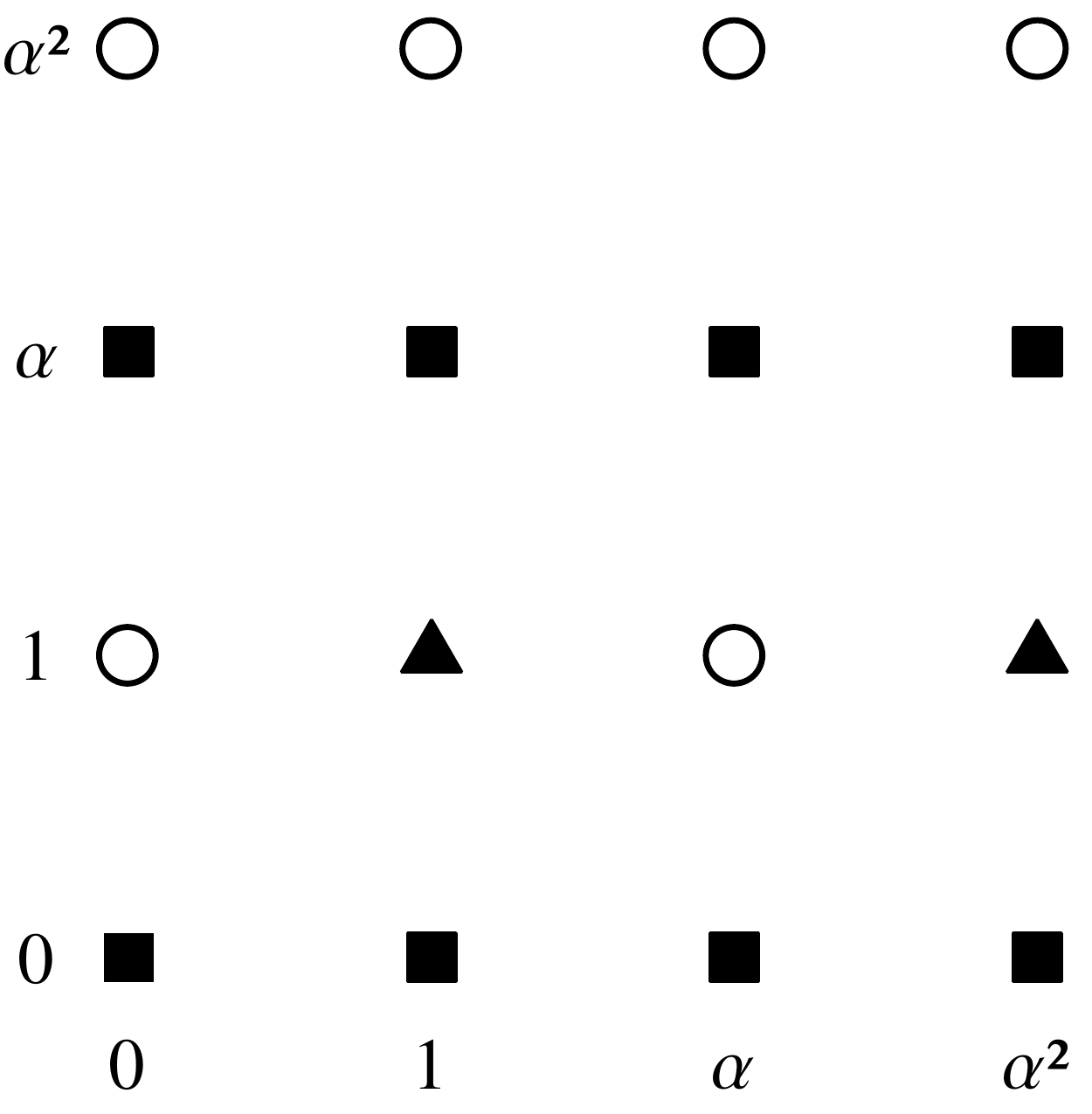}\hfill
  \includegraphics[width=3cm]{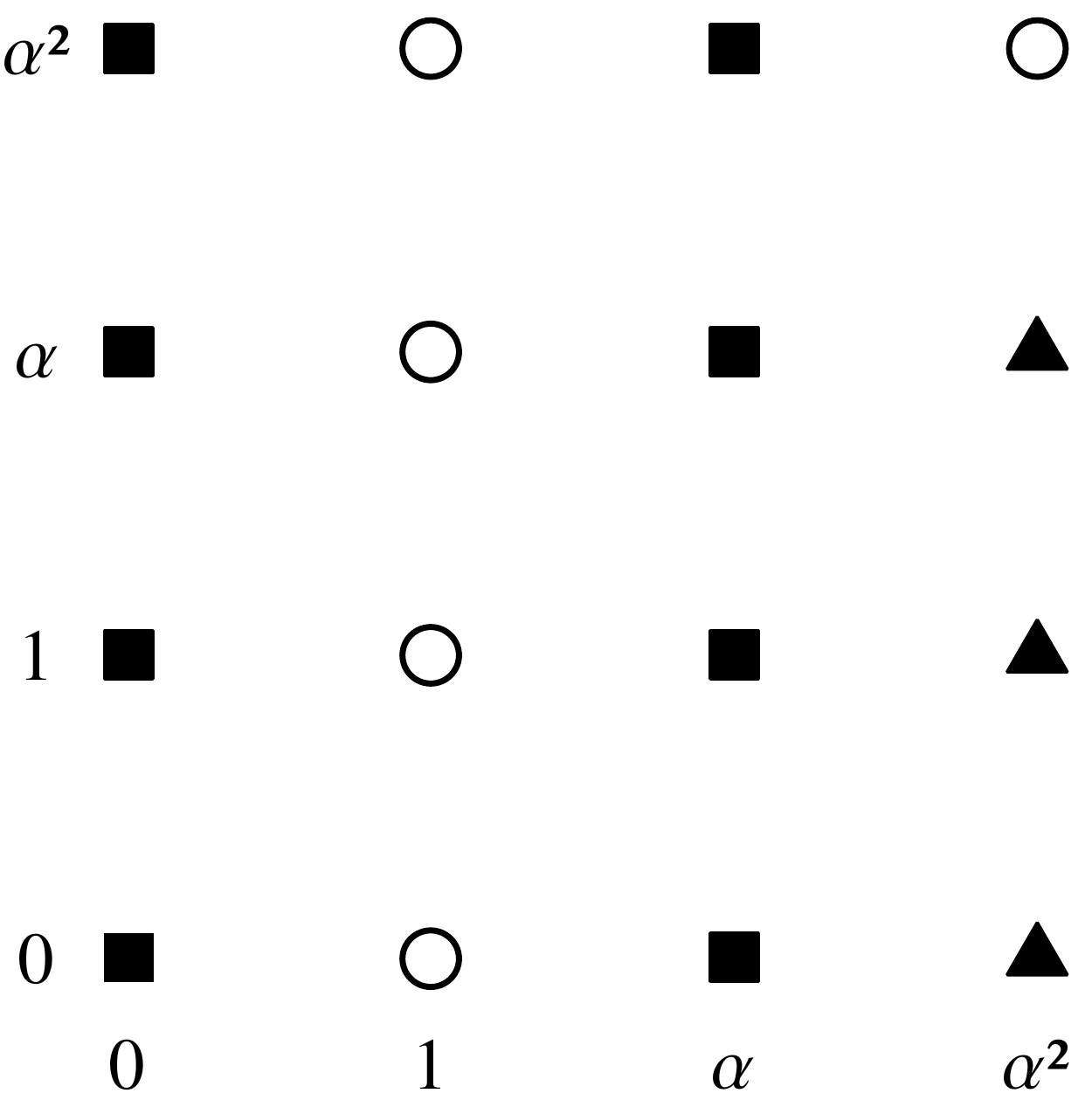}
  \caption{Card deals in $\mathbb F^2_4$}
  \label{twentyfive}
\end{figure}

Why does this bijection inform Bob of Alice's cards? Consider the
configuration in Figure \ref{twentyfive} (left), where
Alice's cards are represented as $\blacksquare$, Bob's as
$\ocircle$ and Cath's as $\blacktriangle$. This two-dimensional plane $\mathbb F^2_4$ consists of 20 lines (hyperplanes), such as the four horizontal and the four vertical lines in the figure. The other three foursomes are somewhat less obvious to visualize based on this representation; for example, the line $x=y$ is given by the set $L=\{\pair 00,\pair 11,\pair \alpha\alpha,\pair{\alpha^2}{\alpha^2}\}$, whereas the parallel line $\pair 0{\alpha^2}+L$ is given by $\{\pair 0{\alpha^2},\pair 1{\alpha},\pair \alpha 1,\pair {\alpha^2}0\}$. Just as for $(3,3,1)$ and the seven-point projective plane, Alice's announcement rules out some possibilities for her hand, as not every set of eight points includes two parallel lines. 

Alice's announcement is sufficient for Bob to
determine Alice's hand. Suppose that $A$, $B$, $C$ are the sets of cards that
Alice, Bob and Cath hold, respectively. Then, $A\cup C$ contains ten
points, and thus cannot contain two different pairs of parallel
lines, as the minimum set to contain different pairs of parellel lines is twelve points. In Figure \ref{twentyfive} (left), $A\cup C$ contains {\em three} lines, namely the two verticals and the one horizontal, but only {\em one} pair of parallel lines. More generally, one can check no set of ten points contains two distinct pairs of parallel lines.

On the other hand, it is not merely the case that Cath, who holds points (cards) $\pair 11$ and $\pair{\alpha^2}1$, cannot determine the ownership of a single card not in her possession, but it is even the case that for every pair $\{x,y\}$ of points, Cath considers it possible that $\{x,y\}\subseteq A$ and that $\{x,y\}\not\subseteq A$. In other words, the protocol is $2$-safe; for example, consider the pair $\{\pair{\alpha^2}0,\pair{\alpha^2}\alpha\}$. In Figure \ref{twentyfive} (left) these cards are held by Bob, but in Figure \ref{twentyfive} (center) these cards are held by Alice. From Cath's perspective, the card deal $(A',B',C)$ depicted in the center, where $\{\pair{\alpha^2}0,\pair{\alpha^2}\alpha\}\subseteq A'$, is indistinguishable from the card deal $(A,B,C)$ on the left, where $\{\pair{\alpha^2}0,\pair{\alpha^2}\alpha\}\not\subseteq A$, even though $\{\pair{\alpha^2}0,\pair{\alpha^2}\alpha\}\cap A = \varnothing$. This can be verified systematically for all pairs and indeed Cath cannot learn any of Alice's cards, or even predict that Alice holds a card from any given pair. 

However, if Cath has three cards then she may learn some of Alice's cards. Consider 
Figure \ref{twentyfive} (right). Cath learns that the point
$\pair{\alpha^2}{\alpha^2}$ belongs to Bob, as there is no pair of parallel lines
in $A\cup B$ such that one of them contains $\pair{\alpha^2}{\alpha^2}$, for any such pair of lines would cross the line $x=\alpha^2$ in two different points, and Cath would hold one of them. If she has even more than three cards this only gets worse; indeed, it is known that no two-step protocol can be safe and informative when Cath holds as many as or more cards than Alice\footnote{Note however that this restriction may be circumvented by using protocols of more than two steps \cite{colouring}.} \cite{albertetal:2005}. 

\section{Protocols and safety} \label{sec.some} 

Before we continue let us present the notions of {\em protocol,} {\em informativity} and {\em safety} we shall use. Throughout this paper, we will assume that $D$ is a fixed, finite set of ``cards''. A {\em card deal} is a partition $(A,B,C)$ of $D$; the deal has {\em size} $(a,b,c)$ if $A$ is an $a$-set, $B$ a $b$-set and $C$ a $c$-set, where by ``$x$-set'' we mean a set of cardinality $x$. We denote the set of $x$-subsets of $Y$ by $Y\choose x$. We think of $A$ as the {\em hand} of Alice, or that Alice {\em holds $A$}; similarly, $B$ and $C$ are the hands of Bob and Cath, respectively. In general we may simply assume that $D=\{1,\hdots,a+b+c\}$, and define ${\rm Deal}(a,b,c)$ to be the set of partitions of $D$ of size $(a,b,c)$. In \cite{swanson:2012} and other papers, an {\em announcement} has been modelled as a set of hands that one of the agents may hold. Thus Alice would announce a subset $\mathcal A$ of $D\choose a$, indicating that $A\in\mathcal A$, and we follow this presentation.

A characteristic assumption of the problem is that there is a secure dealing phase in which the players learn no information about others' cards, encrypted or otherwise. At the beginning of the protocol, players have knowledge of their own cards, the cards contained in $D$, and of the size $(a,b,c)$ of the deal, but nothing more. Thus they are not able to distinguish between different deals where they hold the same hand. We model this by equivalence relations between deals; since from Alice's perspective, $(A,B,C)$ is indistinguishable from $(A,B',C')$, we define $(A,B,C)\stackrel{Alice}\sim(A',B',C')$ if and only if $A=A'$. We may define analogous equivalence relations for Bob and Cath.

{\em Strategies} of length two are defined in \cite{swanson:2012}. These assign a probability distribution to Alice's possible announcements. If the probability distribution is uniform these are called {\em equitable} strategies. We will work exclusively with equitable strategies, and simply call them {\em protocols.} Since probability is distributed evenly among the possible outcomes, we may dispense with probability measures and merely specify a set of possible announcements.

\begin{definition}[Protocol]\label{defprot}
Fix a size $(a,b,c)$ and let $D=\{1,\hdots,a+b+c\}$. A {\em protocol} (for $(a,b,c)$) is a function $\pi$ assigning to each $A\in {D\choose a}$ a non-empty set $\pi(A)\subseteq \mathcal P\left({D\choose a}\right)$ with the property that $A\in\bigcap \pi(A)$.
\end{definition}

Protocols are non-deterministic in principle; Alice may announce any $\mathcal A\in \pi(A)$. Meanwhile, a successful protocol must have two additional properties. The first is that Alice and Bob know each other's cards (and hence the entire deal) after its execution. Note that it is sufficient for Bob to learn the deal since, once Bob knows Alice's hand, he also knows Cath's and may proceed to announce $C$. Because of this, the protocols we present are in principle two-step protocols, even though we only focus on Alice's announcement and leave Bob's second announcement implicit.

\begin{definition}[Informativity] \label{def.informativity}
Given a deal $(A,B,C)$, an announcement $\mathcal A$ is {\em informative} for $(A,B,C)$ if there is only one $X\in\mathcal A$ such that $X\subseteq A\cup C.$

A protocol for $(a,b,c)$ is informative if for every $(A,B,C)\in{\rm Deal}(a,b,c)$, every $\mathcal A\in \pi(A)$ is informative for $(A,B,C)$.\footnote{Our presentation follows that given in \cite{swanson:2012}. Compare to \cite{hvd.studlog:2003,albertetal:2005}, where {\em informative for $(A,B,C)$} is defined as follows: {\em Given a deal $(A,B,C)$, an announcement $\mathcal A$ containing $A$ is {\em informative} for $(A,B,C)$ if for any $A' \in \mathcal A$ and any $(A',B',C')$, there is only one $X\in\mathcal A$ such that $X\subseteq A'\cup C'$.} In principle this is stronger than the definition we give; however, this is remedied by the more general condition of being {\em informative for $(a,b,c)$.} 
}
\end{definition}

The second property we desire from a protocol is that Cath does not gain ``too much information''. How much is too much depends on a parameter we shall usually call $k$ and states that, given $X\in {D\choose k}$, it is possible from Cath's perspective that $X\subseteq A$ and also possible that $X\not\subseteq A$. This is the notion of {\em weak $k$-security} from \cite{swanson:2012}; we simply call it {\em $k$-safety.}

\begin{definition}[$k$-Safety]\label{w1s}
Given a protocol $\pi$ for $(a,b,c)$ and $A\in {D\choose a}$, an announcement $\mathcal A\in \pi(A)$ is {\em $k$-safe} if for every deal $(A,B,C)$ and every non-empty set $X$ with at most $k$ elements such that $X \cap C=\varnothing$ there is
a deal $(A',B',C)$ such that $X\subseteq A'$ and $\mathcal A\in \pi(A')$, as well as a deal $(A'',B'',C)$ such that $X\not\subseteq A''\in\mathcal A$.

The protocol $\pi$ is $k$-safe if every $\mathcal A\in \pi(A)$ is $k$-safe.
\end{definition}

A stronger notion of security is also discussed in \cite{swanson:2012}. Let us use ${\sf Pr}(\cdot|\cdot)$ to denote conditional probability. {\em Weak $k$-security} is equivalent to the statement that, given a non-empty set $X$ with at most $k$ cards such that $X\cap C=\varnothing$,
\[0<{\sf Pr}(X\subseteq A|\mathcal A,C)<1.\]
This probability may, however, be very small or very large. A stronger notion of security would demand that Cath does not gain probabilistic information from the protocol, so that
\[{\sf Pr}(X\subseteq A|\mathcal A,C)={\sf Pr}(X\subseteq A|C).\]
This is called {\em perfect $k$-security} and is similar to the combinatorial axiom CA4 in \cite{atkinsonetal:2009}. The protocols we shall present are not perfectly secure, but they are weakly $k$-secure for some fixed value of $k$.

\section{The geometric protocol}\label{SecGeoProt}

Here we shall give a formal definition of our protocol in the sense of Definition \ref{defprot} and prove that it indeed provides a $k$-safe solution to the generalized Russian cards problem.

The protocol is based on {\em slicings}:

\begin{definition}[Slicing]
Let $p$ be a prime power and $k,d$ positive integers. Say a set $X\subseteq \mathbb F^{d+1}_p$ is a {\em $k$-slicing} if there are a $d$-dimensional subspace $V$ and $x_1,\hdots,x_k\in\mathbb F^{d+1}_p$ such that $X=\bigcup_{i=1}^k (x_i+ V)$ and $x_i-x_j\in V$ if and only if $i=j$.
\end{definition}

In other words, a $k$-slicing is a union of $k$ parallel hyperplanes. Note that $k$-slicings have exactly $kp^d$ elements.

\begin{definition}[The geometric protocol]\label{DefGeoProt}
Fix a size $(a,b,c)$ such that there are integers $d,k$ and a prime power $p$ with $a=kp^{d}$ and $a+b+c=p^{d+1}$. For a bijection $f:D\to\mathbb F^{d+1}_p$, define $\mathcal A_k[f]$ to be the set of all $X\subseteq D$ such that $f(X)$ is a $k$-slicing.

We then define the {\em geometric protocol} $\gamma$ for $(a,b,c)$ (with parameters $p,d,k$) to be given by $\mathcal A\in\gamma(A)$ if and only if $A\in\mathcal A$ and $\mathcal A=\mathcal A_k[f]$ for some bijection $f:D\to\mathbb F^{d+1}_p$.
\end{definition}

Our main objective is to show that the geometric protocol provides a $k$-safe and informative solution to the generalized Russian cards problem, but this requires for the parameters to satisfy certain conditions. Although we defer the proof, let us state our main theorem now:

\begin{theo}\label{TheoMain}
Assume that $a,b,c,p,d,k$ are such that $p$ is a prime power, $a=kp^{d}$, $a+b+c=p^{d+1}$ and
\begin{align}
c&<kp^{d}-k^2p^{d-1},\label{Cond1}\\
\max\{c+k,ck\}&\leq  p.\label{Cond2}
\end{align}

Then, the geometric protocol with parameters $p,d,k$ is $k$-safe and informative for $(a,b,c)$.
\end{theo}

Before we give a proof, we need to give some preliminary results. These will also help elucidate the purpose of \eqref{Cond1} and \eqref{Cond2}. Note that $\max\{c+k,ck\}$ is usually equal to $ck$ except when either $k$ or $c$ is equal to one. We remark that the bounds given by Theorem \ref{TheoMain} are sufficient but not necessary; for example, the assiduous reader will verify that the protocol is $2$-safe for $(10,12,3)$, yet the bounds we give are not satisfied.

Let us begin with a combinatorial lemma about slicings.

\begin{lemm}\label{LemmDecomp}
Let $p$ be a prime power and $k\in[1,p-1]$. If $X,Y\subseteq \mathbb F^{d+1}_p$ are two distinct $k$-slicings then
\[|X\cup Y|\geq \min\{(k+1)p^{d},2kp^{d}-k^2p^{d-1}\}.\]
\end{lemm}

\begin{proof}
Write $X=\bigcup_{i=1}^kU_i$ and $Y=\bigcup_{i=1}^k W_i$ as disjoint unions of parallel hyperplanes. It may be that for certain values of $i,j$ we have that $U_i=W_j$. If this is the case, since $X\not =Y$ there must be some $U=U_{i_\ast}$ such that $U\not=W_j$ for any $j$. Then, since $U$ is parallel to all $W_i$, we have that $U,W_1,\hdots,W_k$ are mutually disjoint and hence
\[(k+1)p^{d}=\left |U\cup\bigcup_{i\leq k}W_i\right|\leq |X\cup Y|.\]

Now assume this is not the case, so that $U_i\not=W_j$ for any $i,j\leq k$. Then, if $i,j\leq k$ we have $|U_i\cap W_j|\leq p^{d-1}$, so that
\begin{align*}
|X\cup Y|&=\left|\bigcup_{i=1}^k U_i\cup \bigcup_{i=1}^k W_i\right|\\
&\geq\sum_{i=1}^k |U_i|+\sum_{i=1}^k |W_i|-\sum_{i,j=1}^k|U_i\cap W_j|\\
&\geq kp^d+kp^d-k^2p^{d-1}.
\end{align*}
The result follows.
\end{proof}

\begin{lemm}\label{LemmInform}
Assume that $a,b,c,p,d,k$ are such that $p$ is a prime power, $a=kp^d$, $a+b+c=p^{d+1}$ and
\[c<\min\{p^d,kp^d-k^2p^{d-1}\}.\]

Then, the geometric protocol with parameters $p,d,k$ is informative for $(a,b,c)$.
\end{lemm}

\begin{proof} We must show that, given a deal $(A,B,C)$ and a bijection $f:D\to\mathbb F^{d+1}_p$, there can be only one $X\in\mathcal A_k[f]$ with $X\subseteq A\cup C$. More generally, we claim that there may only be one $k$-slicing contained in any set $E\subseteq \mathbb F^{d+1}_p$ with at most $a+c$ points; for indeed, by Lemma \ref{LemmDecomp}, if $E$ contains two distinct $k$-slicings then
\[|E|\geq kp^d+\min\{p^d,kp^d-k^2p^{d-1}\}>a+c.\]
But this contradicts the assumption that $|E|=a+c$, and we conclude that $E$ contains only one $k$-slicing, as claimed.
\end{proof}

Now let us turn to $k$-safety. Once again we begin with purely combinatorial preliminaries. The following is an elementary but useful fact.

\begin{lemm}\label{LemmaDE}
Suppose that $c+k\leq p$ and $|C|\leq c$. Let $V$ be any hyperplane of $\mathbb F^{d+1}_p$. Then, there exist $k$ distinct hyperplanes parallel or equal to $V$ and not meeting $C$.
\end{lemm}

\begin{proof}
Each hyperplane has $p^d$ elements, and each $x\in\mathbb F^{d+1}_p$ lies in a unique hyperplane that is parallel to $V$. It follows that there are $p$ hyperplanes parallel or equal to $V$; since $c+k\leq p$, at least $k$ of them do not meet $C$.
\end{proof}

\begin{lemm}\label{LemmSafePre}
Let $E$ be any subset of $\mathbb F^{d+1}_p$ such that $k|E|\leq p$ and let $X$ be a set of at most $k$ points such that $X\cap E=\varnothing$.

Then, there is a hyperplane $V$ such that $(x_i+ V)\cap E=\varnothing$ for all $i\leq k$.
\end{lemm}

\begin{proof}
We will prove the more general claim that if $E$ is any subset of $\mathbb F^{d+1}_p$ such that
\[k|E|< \dfrac{p^{d-e+2}-1}{p-1}\]
and $X=\{x_1,\hdots,x_k\}$ is a set of at most $k$ points such that $X\cap E=\varnothing$, then there is an $e$-dimensional subspace $V$ of $\mathbb F^{d+1}_p$ such that $(x_i+ V)\cap E=\varnothing$ for all $i\leq k$. The lemma follows by setting $e=d$ and noting that
\[\dfrac{p^{d-d+2}-1}{p-1}=\dfrac{p^{2}-1}{p-1}=p+1.\]

Suppose that $X\subseteq \{x_1,\hdots,x_k\}$. We proceed to build $V$ by induction on $e$. The base case, when $e=0$, is trivial (just take $V=\{0\}$).

For the inductive step, assume that
\[k|E|<\dfrac{p^{d-(e+1)+2}-1}{p-1}<\dfrac{p^{d-e+2}-1}{p-1}.\]
Then, by induction hypothesis there is an $e$-dimensional subspace $V'$ such that for all $i\leq k$, $(x_i+ V')\cap E=\varnothing$. Let us construct an $e+1$-subspace $V\supseteq V'$ with the same property.

Let $\mathcal T$ be the set of all $(e+1)$-dimensional subspaces $U$ such that $V'\subseteq U$. Each $U$ is of the form $\langle u,V'\rangle$ for some $u\not\in V'$. Hence there are
\[\dfrac{p^{d+1}-p^e}{p^{e+1}-p^e}=\dfrac{p^{d-(e+1)+2}-1}{p-1}\]
values that $U$ may take; this is because $u$ may take $p^d-p^e$ different values, but given $u$ we have that $u+ V'=w+ V'$ if and only if $w\in (u+ V')\setminus V'$, and there are $p^{e+1}-p^e$ such $w$.

Meanwhile, if $U\not=W\in\mathcal T$ then $U\cap W=V'$, from which it follows that, given $i\leq k$, $E$ is the disjoint union of all sets of the form $(x_i+ U)\cap E$ with $U\in\mathcal T$. We conclude that for each $i\leq k$, there are at most $|E|$ values of $U\in\mathcal T$ such that $(x_i+ U)\cap E\not=\varnothing$, and hence there are at most $k|E|$ values of $U\in\mathcal T$ such that there exists $i\leq k$ with $(x_i+ U)\cap E\not =\varnothing$. But
\[k|E|< \dfrac{p^{d-(e+1)+2}-1}{p-1},\]
so there is $U_\ast\in\mathcal T$ such that $(x_i+ U_\ast)\cap E=\varnothing$ for all $i\leq k$ and we choose $V={U_\ast}$.
\end{proof}

With this, we may prove that our protocol is safe.

\begin{lemm}\label{LemmSafe}
Assume that $a,b,c,p,d,k$ are such that $p$ is a prime power, $a=kp^d$, $a+b+c=p^{d+1}$ and further \eqref{Cond2} holds.

Then, the geometric protocol with parameters $p,d,k$ is $k$-safe for $(a,b,c)$.
\end{lemm}

\begin{proof}
Suppose that the deal $(A,B,C)$ is given and Alice has announced $\mathcal A=\mathcal A_k[f]\in\gamma(A)$.

Choose a set $X\subseteq D$ that has at most $k$ elements with $X\cap C=\varnothing$. By Lemma \ref{LemmSafePre}, there is a $d$-dimensional subspace $V$ such that $(x+ V)\cap f(C)=\varnothing$ for all $x\in f(X)$. Let $A_1',\hdots,A_m'$ be all sets of the form $x+ V$ with $x\in f(X)$; we know that $m\leq k$, but note that it may be the case that $m< k$. However, in view of Lemma \ref{LemmaDE}, there are at least $k$ different hyperplanes parallel or equal to $V$ and not meeting $f(C)$ and thus we may pick $A_{m+1}',\hdots,A_k'$ parallel or equal to $V$ but distinct from $A_i'$ for $i\leq m$. Setting $A'=\bigcup _{i=1}^kA_i'$ we see that $A'\in \mathcal A$ is a $k$-slicing containing $X$ and not meeting $f(C)$, so that $f^{-1}(A')\in\mathcal A_k[f]$ and $f^{-1}(A')\cap C=\varnothing$, as required.

Meanwhile, to find an element of $\mathcal A_k[f]$ not containing $X$ and not meeting $C$, it suffices to find a $k$-slicing $A''$ not containing any fixed $x\in f(X)$ and not meeting $f(C)$. Choose any $y\in f(C)$ and any $d$-dimensional subspace $W$ such that $x-y\in W$. Once again use Lemma \ref{LemmaDE} to pick $A''_1,\hdots,A''_k$ parallel or equal to $W$ and not meeting $f(C)$ and set $A''=\bigcup_{i=1}^k A_i''$. Then, $A''$ is a $k$-slicing not meeting $f(C)\cup \{x\}$, which means that $f^{-1}(A'')$ does not meet $C$ and does not contain $X$, as required.
\end{proof}

Our main result is now immediate:

\begin{proof}[Proof of Theorem \ref{TheoMain}]
Suppose $a,b,c,p,d,k$ satisfy \eqref{Cond1} and \eqref{Cond2}. Note that by \eqref{Cond2}, $c+k\leq p$ and thus $c<p^d$; it follows that $c< \min\{p^d,kp^d-k^2p^{d-1}\}$. Then by Lemma \ref{LemmInform}, the geometric protocol is informative, whereas by Lemma \ref{LemmSafe} it is $k$-safe.
\end{proof}

\section{Computing parameters}\label{SecCompPat}

Theorem \ref{TheoMain} gives general conditions under which the geometric protocol works, but it is perhaps not obvious how to find suitable parameters or even that many exist. In this section we shall flesh out more specific consequences of this result, showcasing its usefulness in solving many new instances of the Russian cards. We remark, however, that the bounds we give here are not meant to be exhaustive; the different parameters can be chosen in many other ways to obtain solutions for deals of different sizes.

First let us give a simplified version of \eqref{Cond1}, which will be easier to work with:

\begin{lemm}\label{Lemma12}
Given natural numbers $p,k\geq 1$ we have that $kp^d-k^2p^{d-1}> 0$ if and only if $k\in[1,p-1]$, in which case $kp^d-k^2p^{d-1}\geq p^d-p^{d-1}$.
\end{lemm}

\begin{proof}
The function $kp^d-k^2p^{d-1}$ is concave on $k$ and $kp^d-k^2p^{d-1}=0$ when $k=0$ or $k=p$. It follows that, for natural $k$, $kp^d-k^2p^{d-1}> 0$ if and only if $1\leq k\leq p-1$.

Now, when $k=1$ we have that $kp^d-k^2p^{d-1}=p^d-p^{d-1},$ and similarly when $k=p-1$, from which it follows once again by concavity that $kp^d-k^2p^{d-1}\geq p^d-p^{d-1}$ for all $k\in[1,p-1]$.
\end{proof}

Before we continue, let us mention a simple number-theoretic observation which will nevertheless be very useful:

\begin{lemm}\label{LemmaIsP}
Given $n\geq 1$ there exists a prime power $p$ such that $n<p \leq 2n$.
\end{lemm}

\begin{proof}
Choose $\ell$ to be the unique integer such that $n<2^\ell\leq 2n $ and set $p=2^\ell$.
\end{proof}

Now for the main result of this section. In order to give a uniform bound we shall use the fact that for all $k,c\geq 1$ we have that $\max\{k+c,kc\}\leq kc+1$.

\begin{theo}\label{TheoKArbitrary}
Given $k,c\geq 1$. Then, the geometric protocol is informative and $k$-safe for $(a,b,c)$ for infinitely many values of $a,b$ with $b<2a(c+1)$. The smallest such value of $a$ is at most $2k(kc+1)$.
\end{theo}

\begin{proof} As before, it suffices to check that all conditions of Theorem \ref{TheoMain} may be satisfied for appropriately chosen parameters.

Fix $k\geq 1$ and $c>0$. Use Lemma \ref{LemmaIsP} to find a prime power $p$ such that $kc+1<p\leq 2(kc+1)$. Given $d\geq 1$, set $a=kp^d$ and $b=p^{d+1}-a-c$.

First we observe that if $d=1$, $a=kp\leq 2k(kc+1)$. Note that $b<\frac{pa}k\leq 2a(c+1)$ independently of $d$.

To see that condition \eqref{Cond1} holds, it suffices in view of Lemma \ref{Lemma12} to observe that
\[c<p-1 \leq p^{d-1}(p-1)=p^d-p^{d-1}.\]
Meanwhile, condition \eqref{Cond2} holds by the way we chose $p$.

It follows from Theorem \ref{TheoMain} that the geometric protocol is informative and $k$-safe for $(a,b,c)$, as claimed.
\end{proof}

\section{Conclusions and future work}

The generalized Russian cards problem, aside from being interesting from a purely combinatorial perspective, provides a prototypical case-study for {\em in\-for\-ma\-tion-based cryptography} \cite{maurer:1999}. This alternative paradigm is based on the {\em impossibility,} rather than the {\em improbability,} of encrypted messages being intercepted. Information-based cryptography has the advantage over probability-based cryptography that it does not depend on eavesdroppers' computational resources or, for example, the assumption that $\text{\sc P}\neq\text{\sc NP}$. However, such methods are less developed, and for good reason, as the demand that protocols be {\em unconditionally secure} is rather strong. Also, they depend on a prior phase of key-distribution (`dealing the cards') by some central authority, that is assumed to be perfectly secure. Nevertheless, solutions to the generalized Russian cards problem could very well lead, directly or indirectly, to applications in information-theoretic methods of secure communication.

In this paper we have given improved bounds for the tuple $(a,b,c)$ to have an informative and $k$-safe solution. The notion of $k$-safety was originally considered by Stinson and Swanson in \cite{swanson:2012}, where $k$-safe protocols are given for many cases where Cath has one card. Perhaps the main contribution of the present work is that the protocol we present gives the first $k$-safe solutions for $c>1$.

Stinson and Swanson also introduce {\em perfect} security. This strengthened notion of security may yet be extended to many new cases, for in the methods we propose (and most that are available in the literature), there is no guarantee that Cath does not learn additional probabilistic information. It may be possible to amend the geometric protocol to achieve either perfect security or, more likely, an approximate variant, where the probabilistic information that Cath learns is minimal.

Finally, the authors have presented in \cite{colouring} a method which allows Cath to have more cards than Alice, which we call the {\em colouring protocol.} It is known that such a protocol must have more than two steps; ours has four. In principle, the present geometric protocol may be combined with the colouring protocol to give $k$-safe solutions when Cath has many more cards than Alice, although the combinatorial analysis is likely to be rather challenging.

These are possible avenues to explore in future work; however, the Russian cards problem allows for a large degree of freedom and it may very well be that entirely new and better methods shall be developed to obtain more efficient and secure protocols.

\end{document}